\documentclass[10pt]{article}

\usepackage{fullpage}
\usepackage{amsmath}
\usepackage{amsthm}
\usepackage{amssymb}
\usepackage{amsfonts}
\usepackage{hyperref}

\newcommand{\ignore}[1]{}

\newtheorem{theorem}{Theorem}

\newtheorem{lemma}[theorem]{Lemma}
\newtheorem{fact}{Fact}

\renewcommand{\Pr}{{\bf Pr}}
\newcommand{\E}{{\bf E}}
\newcommand{\D}{{\cal D}}

\begin{document}

\title{Almost Optimal Distribution-free Junta Testing}
\author{{\bf Nader H. Bshouty}\\ Dept. of Computer Science\\ Technion,  Haifa, 32000\\
}

\maketitle
\begin{abstract}
We consider the problem of testing whether an unknown $n$-variable Boolean function is a $k$-junta in the distribution-free property testing model, where the distance between functions is measured with respect to an arbitrary and unknown probability distribution over $\{0,1\}^n$. Chen, Liu, Servedio, Sheng and Xie~\cite{LiuCSSX18} showed that the distribution-free $k$-junta testing can be performed, with one-sided error, by an adaptive algorithm that makes $\tilde O(k^2)/\epsilon$ queries. In this paper, we give a simple two-sided error adaptive algorithm that makes $\tilde O(k/\epsilon)$ queries.
\end{abstract}

\section{Inroduction}
Property testing of Boolean function was first considered in the seminal works of Blum, Luby and Rubinfeld~\cite{BlumLR93} and Rubinfeld and Sudan~\cite{RubinfeldS96} and has recently become a very active research area. See for example,~\cite{AlonKKLR05,BaleshzarMPR16,BelovsB16,BhattacharyyaKSSZ10,BlaisBM11,BlaisK12,ChakrabartyS13,ChakrabartyS16a,ChenDST15,ChenST14,ChenWX17,ChenWX17b,DiakonikolasLMORSW07,FischerKRSS04,
GoldreichGLRS00,GopalanOSSW11,KhotMS15,KhotS16,MatulefORS10,MatulefORS09,ParnasRS02,Saglam18} and other works  referenced in the surveys~\cite{GoldreichSurvey10,Ron08,Ron09}.

A function $f:\{0,1\}^n\to \{0,1\}$ is said to be $k$-junta if it depends on at most $k$ variables. Juntas have been of particular interest to the computational learning theory community~\cite{Blum03,BlumL97,BshoutyC18,GuijarroTT98,LiptonMMV05,MosselOS04}. A problem closely related to learning juntas is the problem of testing juntas: Given black-box query access to a Boolean function $f$. Distinguish, with high probability, the case that $f$ is $k$-junta versus the case that $f$ is $\epsilon$-far from every $k$-junta.

In the uniform distribution framework, where the distance between two functions is measured with respect to the uniform distribution, Ficher et al.~\cite{FischerKRSS04} introduced the junta testing problem and gave adaptive and non-adaptive algorithms that make $poly(k)/\epsilon$ queries. Blais in~\cite{Blais08} gave a non-adaptive algorithm that makes $\tilde{O}(k^{3/2})/\epsilon$ queries and in~\cite{Blais09} an adaptive algorithm that makes $O(k\log k+k/\epsilon)$ queries.
On the lower bounds side, Fisher et al.~\cite{FischerKRSS04} gave an $\Omega(\sqrt{k})$ lower bound. Chockler and Gutfreund~\cite{ChocklerG04} gave an $\Omega(k)$ lower bound for adaptive testing and, recently, Sa\u{g}lam in~ \cite{Saglam18} improved this lower bound to $\Omega(k\log k)$. For the non-adaptive testing Chen et al.~\cite{ChenSTWX17} gave the lower bound $\tilde{\Omega}(k^{3/2})/\epsilon$.

In the {\it distribution-free property testing},~\cite{GoldreichGR98}, the distance between Boolean functions is measured with respect to an arbitrary and unknown distribution~${\cal D}$ over $\{0,1\}^n$. In this model, the testing algorithm is allowed (in addition to making black-box queries) to draw random $x\in\{0,1\}^n$ according to the distribution~${\cal D}$. This model is studied in~\cite{ChenX16,DolevR11,GlasnerS09,HalevyK07,LiuCSSX18}. For testing $k$-junta in this model, Chen et al.~\cite{LiuCSSX18} gave a one-sided adaptive algorithm that makes $\tilde{O}(k^{2})/\epsilon$ queries and proved a lower bound $\Omega(2^{k/3})$ for any non-adaptive algorithm. The results of Halevy and Kushilevitz~\cite{HalevyK07} gives a one-sided non-adaptive algorithm that makes $O(2^k/\epsilon)$ queries. The adaptive $\Omega(k\log k)$ uniform-distribution lower bound from~\cite{Saglam18} trivially extend to the distribution-free model.

In this paper, we close the gap between the adaptive lower and upper bound. We prove

\begin{theorem} For any $\epsilon>0$, there is a two-sided distribution-free adaptive algorithm for $\epsilon$-testing $k$-junta that makes $\tilde O(k/\epsilon)$ queries.
\end{theorem}

Our exact upper bound is $O((k/\epsilon)\log (k/\epsilon))$ and therefore, by Sa\u{g}lam~\cite{Saglam18} lower bound of $\Omega(k\log k)$, our bound is tight for any constant $\epsilon$.

\section{Preliminaries}
In this section we give some notations follows by a formal definition of the model and some preliminary known results
\subsection{Notations}\label{Notation}
We start with some notations. Denote $[n]=\{1,2,\ldots,n\}$. For $S\subseteq [n]$ and $x=(x_1,\ldots,x_n)$ we write $x(S)=\{x_i|i\in S\}$. For $X\subset [n]$ we denote by $\{0,1\}^X$
the set of all binary strings of
length $|X|$ with coordinates indexed by $i\in X$. For $x\in \{0,1\}^n$ and $X\subseteq [n]$ we write $x_X\in\{0,1\}^{X}$ to denote the projection of $x$ over coordinates in $X$. We denote by $1_X$ and $0_X$ the all one and all zero strings in $\{0,1\}^{X}$, respectively. When we write $x_I=0$ we mean $x_I=0_I$. For $X_1,X_2\subseteq [n]$ where $X_1\cap X_2=\emptyset$ and $x\in \{0,1\}^{X_1}, y\in \{0,1\}^{X_2}$ we write $x\circ y$  to denote their concatenation,
the string in $\{0,1\}^{X_1\cup X_2}$ that agrees with $x$ over coordinates in $X_1$ and agrees with $y$ over~$X_2$. For $X\subseteq [n]$ we denote $\overline{X}=[n]\backslash X$. We say that the Boolean function $f:\{0,1\}^n\to \{0,1\}$ is a literal if $f\in \{x_1,\ldots,x_n,\bar{x_1},\ldots,\bar{x_n}\}$.

Given $f,g:\{0,1\}^n\to \{0,1\}$ and a probability distribution $\D$ over $\{0,1\}^n$, we say that $f$ is $\epsilon$-{\it close to $g$ with respect to} $\D$ if $\Pr_{x\in\D}[f(x)\not=g(x)]\le \epsilon$, where $x\in \D$ means $x$ is chosen from $\{0,1\}^n$ according to the distribution $\D$. We say that $f$ is $\epsilon$-{\it far from $g$ with respect to} $\D$ if $\Pr_{x\in\D}[f(x)\not=g(x)]\ge \epsilon$. We say that $f$ is $\epsilon$-{\it far from every $k$-junta with respect to} $\D$ if for every $k$-junta $g$, $f$ is $\epsilon$-far from $g$ with respect to $\D$. We will use $U$ to denote the uniform distribution over $\{0,1\}^n$.

\subsection{The Model}
In this subsection, we define the model.

We consider the problem of testing juntas in the distribution-free testing model. In this model, the algorithm has access to a $k$-junta $f$ via a black-box that returns $f(x)$ when a string $x$ is queried, and access to unknown distribution $\D$ via an oracle that returns $x\in\{0,1\}^n$ chosen randomly according to the distribution $\D$.

A {\it distribution-free testing algorithm} ${\cal A}$ is a algorithm that, given as input a distance parameter $\epsilon$ and the above two oracles,
\begin{enumerate}
\item if $f$ is $k$-junta then ${\cal A}$ output ``accept'' with probability at least $2/3$.
\item if $f$ is $\epsilon$-far from every $k$-junta with respect to the distribution $\D$ then it output ``reject'' with probability at least $2/3$.
\end{enumerate}

We say that ${\cal A}$ is {\it one-sided} if it always accepts when $f$ is $k$-junta, otherwise, it is called {\it two sided} algorithm. The {\it query complexity} of a distribution-free testing algorithm is the number of queries made on $f$.

\subsection{Preliminaries Results}\label{PRE}
In this section, we give some known results that will be used in the sequel.

For a Boolean function $f$ and $X\subset [n]$, we say that $X$ is a {\it relevant set} of $f$ if there
are $a,b\in \{0,1\}^n$ such that $f(a)\not= f(b_X\circ a_{\overline X})$. When $X=\{i\}$ then we say that $x_i$ is {\it relevant variable} of $f$. Obviously, if $X$ is relevant set of $f$ then $x(X)$ contains at least one relevant variable of $f$. In particular, we have
\begin{lemma}\label{trivial01} If $\{X_i\}_{i\in [r]}$ is a partition of $[n]$ then for any Boolean function $f$ the number of relevant sets $X_i$ of $f$ is at most the number of relevant variables of $f$.
\end{lemma}

We will use the following folklore result that is formally proved in~\cite{LiuCSSX18}.
\begin{lemma}\label{BiSe} Let $\{X_i\}_{i\in [r]}$ be a partition of $[n]$. Let $f$ be a Boolean function and $u\in \{0,1\}^n$. If $f(u)\not= f(0)$ then a relevant set $X_\ell$ of $f$ with a string $v\in \{0,1\}^n$ that satisfies $f(v)\not=f(0_{X_\ell}\circ v_{\overline{X_\ell}})$ can be found with $\lceil \log_2 r\rceil$ queries.
\end{lemma}

The following is from \cite{Blais09}
\begin{lemma}\label{OneSide} There exists a one-sided adaptive algorithm, {\bf UniformJunta}$(f , k, \epsilon,\delta)$, for $\epsilon$-testing $k$-junta that makes $O(((k/\epsilon) + k \log k)\log(1/\delta))$ queries and rejects $f$ with probability at least $1-\delta$ when it is $\epsilon$-far from every $k$-junta with respect to the uniform distribution.
\end{lemma}

The following is from \cite{LiuCSSX18}.

\begin{lemma}\label{disD} Let $\D$ be any probability distribution over $\{0,1\}^n$. If $f$ is $\epsilon$-far from every $k$-junta with respect to $\D$ then for any $J\subseteq [n]$, $|J|\le k$ we have
$$\Pr_{x\in \D,y\in U}[f(x)\not =f(x_J\circ y_{\bar{J}})]\ge \epsilon.$$
\end{lemma}
\begin{proof} Let $J\subseteq [n]$ of size $|J|\le k$.
For every fixed $y\in \{0,1\}^n$ the function $f(x_J\circ y_{\bar{J}})$ is $k$-junta and therefore
$\Pr_{x\in \D}[f(x)\not=f(x_J\circ y_{\bar{J}})]\ge \epsilon.$ Therefore $$\Pr_{x\in \D,y\in U}[f(x)\not =f(x_J\circ y_{\bar{J}})]\ge \epsilon.$$
\end{proof}

\section{The Algorithm}

In this section, we prove the correctness of the algorithm and show that it makes $\tilde O(k/\epsilon)$ queries. We first give an overview of the algorithm then prove its correctness and analyze its query complexity.

\subsection{Overview of the Algorithm}
In this subsection we give an overview of the algorithm. We will use the notationד in Subsection~\ref{Notation} and the definitions and Lemmas in Subsection~\ref{PRE}.

Consider the algorithm in Figure~\ref{A3}. In steps~\ref{par1}-\ref{par2}, the algorithm uniformly at random partitions $[n]$ into $r=2k^2$ disjoint sets $X_1,\ldots,X_r$. Lemma~\ref{dist} shows that,
\begin{fact} \label{fact01} If the function is $k$-junta then with high probability (w.h.p), each set of variables $x(X_i)=\{x_j|j\in X_i\}$ contains at most one relevant variable.
\end{fact}
In steps~\ref{Sett}-\ref{EndRep}, the algorithm finds
\begin{fact}\label{fact02}
relevant sets $\{X_i\}_{i\in I}$ such that for $X=\cup_{i\in I}X_i$, w.h.p., the function $f(x_X\circ 0_{\overline{X}})$ is $\epsilon/2$-close to $f$ with respect to $\D$.
\end{fact}
To find such set, the algorithm, after finding relevant sets $\{X_i\}_{i\in I'}$, chooses random string $u\in \D$ and tests if $f(u_{X'}\circ 0_{\overline{X'}})\not= f(u)$ where $X'=\cup_{i\in I'}X_i$. The variable $t(X')$ counts for how many random strings $u\in \D$ we get $f(u_{X'}\circ 0_{\overline{X'}})= f(u)$. If $t(X')$ reaches the value $O((\log k)/\epsilon)$ then, w.h.p, $f(x_{X'}\circ 0_{\overline{X'}})$ is $\epsilon/2$-close to $f$ with respect to $\D$ and $X=X'$. Otherwise, $f(u_{X'}\circ 0_{\overline{X'}})\not= f(u)$ and using Lemma~\ref{BiSe} the algorithm finds a new relevant set $X_\ell$.
This is proved in Lemma~\ref{cloose}.

In addition, for each relevant set $X_\ell$, $\ell\in I$, it finds a string $v^{(\ell)}$ that satisfies $f(v^{(\ell)})\not= f(0_{X_\ell}\circ v^{(\ell)}_{\overline{X_\ell}})$. Obviously, if $|I|>k$ then, since each relevant set contains at least one relevant variable, the target is not $k$-junta and the algorithm rejects. See Lemma~\ref{trivial01}.

Now one of the key ideas is the following: If $f$ is $k$-junta then $f(x_X\circ 0_{\overline{X}})$ is $k$-junta. If $f$ is $\epsilon$-far from every $k$-junta with respect to $\D$ then since, by Fact~\ref{fact02}, w.h.p., $f(x_X\circ 0_{\overline{X}})$ is $\epsilon/2$-close to $f$ with respect to $\D$ we have that,
\begin{fact}\label{fact03}
If $f$ is $\epsilon$-far from every $k$-junta with respect to $\D$ then, w.h.p.,
$f(x_X\circ 0_{\overline{X}})$ is $\epsilon/2$-far from every $k$-junta with respect to $\D$.
\end{fact}

Now, since each $X_\ell$, $\ell\in I$ is relevant set and $f(v^{(\ell)})\not= f(0_{X_\ell}\circ v^{(\ell)}_{\overline{X_\ell}})$, for $\ell\in I$ the function $f(x_{X_\ell}\circ v^{(\ell)}_{\overline{X_\ell}})$ is non-constant.
In steps~\ref{LitT}-\ref{ConE}, the algorithm tests that,
\begin{fact}\label{fact04}
w.h.p., for each $\ell\in I$ there is $\tau(\ell)\in X_\ell$ such that $f(x_{X_\ell}\circ v^{(\ell)}_{\overline{X_\ell}})$ is close to some literal in $\{x_{\tau(\ell)},\overline{x_{\tau(\ell)}}\}$, with respect to the uniform distribution.
\end{fact}
This is done using the procedure {\bf UniformJunta} in Lemma~\ref{OneSide}.

If $f$ is $k$-junta then, by Fact~\ref{fact01} and~\ref{fact02}, w.h.p., it passes this test (does not output reject). This is Lemma~\ref{kjun}. If the algorithm does not pass this test, it rejects. If $f$ is not $k$-junta and it passes this test, then the statement in Fact~\ref{fact04} is true. This is proved in Lemma~\ref{closelit}.

Consider now steps~\ref{RepF}-\ref{Finn}. First, let us consider a function $f$ that is $\epsilon$-far from every $k$-junta with respect to $\D$. Let $J=\{\tau(\ell)\ |\ \ell\in I\}$ where $\tau(\ell)$ is as defined in Fact~\ref{fact04}. Since by Fact~\ref{fact03}, w.h.p., $f(x_X\circ 0_{\overline{X}})$ is $\epsilon/2$-far from every $k$-junta with respect to $\D$ and $|J|=|I|\le k$, by Lemma~\ref{disD}, w.h.p.,
$$\Pr_{y\in U,x\in {\cal D}}[f(x_X\circ 0_{\overline{X}})\not=f(x_J\circ y_{X\backslash J}\circ 0_{\overline{X}})]\ge \epsilon/2.$$
So we need to test whether $f(x_X\circ 0_{\overline{X}})$ is $\epsilon/2$-far from $f(x_J\circ y_{X\backslash J}\circ 0_{\overline{X}})$ (those are equal in the case when $f$ is $k$-Junta).
This is the last test we would like to do but the problem is that we do not know $J$, so we cannot use this test as is. So we change it, as is done in~\cite{LiuCSSX18}, to an equivalent test as follows
$$\Pr_{z\in U,x\in {\cal D}}[f(x_X\circ 0_{\overline{X}})\not=f((x_X+z_X)\circ 0_{\overline{X}})\  |\  z_{J}=0_J]\ge \epsilon/2.$$
To be able to draw uniformly random $z_X$ with $z_J=0_J$, we use Fact~\ref{fact04}, that is, the fact that each $f(x_{X_\ell}\circ v_{\overline{X_\ell}}^{(\ell)})$ is close to one of the literals in $\{x_{\tau(\ell)},\overline{x_{\tau(\ell)}}\}$. For every $\ell\in I$, the algorithm draws uniformly random $w:=z_{X_\ell}$ and then using the fact that $f(x_{X_\ell}\circ v_{\overline{X_\ell}}^{(\ell)})$ is close to one of the literals in $\{x_{\tau(\ell)},\overline{x_{\tau(\ell)}}\}$ where $\tau(\ell)\in X_\ell$ the algorithm tests in which set $Y_{\ell,0}:=\{j\in X_\ell |\ w_j=0\}$ or $Y_{\ell,1}:=\{j\in X_\ell |\ w_j=1\}$ the index $\tau(\ell)$ falls. If $\tau(\ell)\in Y_{\ell,0}$ then the entry $\tau(\ell)$ in $z_{X_\ell}$ is zero and if $\tau(\ell)\in Y_{\ell,1}$ then the entry $\tau(\ell)$ in $z_{X_\ell}$ is one. In the latter case, the algorithm replaces $z_{X_\ell}$ with $\overline{z_{X_\ell}}$ (negation of each entry in $z_{X_\ell}$) which is also uniformly random.  This gives a random uniform $z_{X_\ell}$ with $z_{\tau(\ell)}=0$. We do that for every $\ell\in I$ and get a random uniform $z$ with $z_{J}=0$. This is proved in Lemma~\ref{FinTest}. Then the algorithm rejects if $f(x_X\circ 0_{\overline{X}})\not=f((x_X+z_X)\circ 0_{\overline{X}})$. If $f(x_X\circ 0_{\overline{X}})$ is $\epsilon/2$-far from every $k$-junta then, by Lemma~\ref{disD}, $f(x_X\circ 0_{\overline{X}})$ is $\epsilon/2$-far from $f(x_J\circ y_{X\backslash J}\circ 0_{\overline{X}})$, and the algorithm, with one test, rejects with probability at least $\epsilon/2$. Therefore, by repeating this test $O(1/\epsilon)$ times the algorithm rejects w.h.p. This is proved in Lemma~\ref{kll}.

Now we consider $f$ that is $k$-junta. Obviously, if $f$ is $k$-junta then $f(x_X\circ 0_{\overline{X}})=f((x_X+z_X)\circ 0_{\overline{X}})$ when $z_J=0$ and the algorithm accepts. This is because $x(J)$ are the relevant variables in $f(x_X\circ 0_{\overline{X}})$. This is proved in Lemma~\ref{kjun2}.

\newcounter{ALC}
\setcounter{ALC}{0}
\newcommand{\step}{\stepcounter{ALC}$\arabic{ALC}.\ $\>}
\newcommand{\steplabel}[1]{\addtocounter{ALC}{-1}\refstepcounter{ALC}\label{#1}}
\begin{figure}[h!]
  \begin{center}
  \fbox{\fbox{\begin{minipage}{28em}
  \begin{tabbing}
  xxx\=xxxx\=xxxx\=xxxx\=xxxx\=xxxx\= \kill
  {{\bf Algorithm Simple${\cal D}k-$Junta}$(f,\D,\epsilon)$}\\
 {\it Input}: Oracle that accesses a Boolean function $f$ and \\
\>\>oracle that draws a random $x\in \{0,1\}^n$ according to the distribution $\D$. \\
  {\it Output}: Either ``accept'' or ``reject''\\ \\
{\bf Partition $[n]$ into $r$ sets}\\
\step\steplabel{par1}
Set $r = 2k^2$.\\
\step\steplabel{par2}
Choose uniformly at random a partition $X_1,X_2,\ldots,X_r$ of $[n]$\\
\\
{\bf Find a close function and relevant sets} \\
\step\steplabel{Sett}
Set $X=\emptyset$; $I=\emptyset$; $t(X)=0$\\
\step\steplabel{two}
Repeat $M=2k\ln(15k)/\epsilon$ times\\
\step\steplabel{Cho}
\> Choose $u\in {\cal D}$. \\
\step  \> $t(X)\gets t(X)+1$\\
\step\steplabel{con1}
\> If $f(u_X\circ 0_{\overline{X}})\not=f(u)$ then\\
\step\steplabel{Find}
\>\>\> Binary search to find a new relevant set $X_\ell$; $X\gets X\cup X_\ell$; $I\gets I\cup \{\ell\}$\\
\step\steplabel{Finddd}    \>\>\>\> and a string $v^{(\ell)}\in\{0,1\}^n$ such that $f(v^{(\ell)})\not= f(0_{X_\ell}\circ v^{(\ell)}_{\overline{X_\ell}})$.\\
\step\steplabel{Rej}
\>\>\> If $|I|>k$ then output ``reject'' and halt.\\
\step\steplabel{tx0}
\>\>\> $t(X)=0$.\\
\step\steplabel{EndRep}
\>  If $t(X)=2\ln(15k)/\epsilon$ then Goto \ref{LitT}.\\ \\
{\bf Tests if each relevant set corresponds to a Boolean function that is close to a literal}\\
\step\steplabel{LitT}
For every $\ell\in I$ do\\
\step\steplabel{Uni}
\>  If {\bf UniformJunta}$(f(x_{X_\ell}\circ v^{(\ell)}_{\overline{X_\ell}}),1,1/30,1/15)$=``reject'' \\
\step\steplabel{Rej2}
\>\>\>then output ``reject'' and halt\\
\step\steplabel{ConB}
\> Choose $b\in U$\\
\step\steplabel{ConE}
\> If $f(b_{X_\ell}\circ v^{(\ell)}_{\overline{X_\ell}})=f(\overline{b_{X_\ell}}\circ v^{(\ell)}_{\overline{X_\ell}})$ then output ``reject'' and halt\\ \\
{\bf The final test of Lemma~\ref{disD}} \\
\step\steplabel{RepF}
Repeat $M'=(2\ln 15)/\epsilon$ times\\
\step\steplabel{feld}
\>  Choose $w\in U$; $z=0_{\overline{X}}$\\
\step\steplabel{feld01}
\>  For every $\ell\in I$ do\\
\step  \>\>\>  Set $Y_{\ell,\xi}=\{j\in X_\ell |  w_j=\xi\}$ for $\xi\in\{0,1\}$.\\
\step  \>\>\>  Set $G_{\ell,0}=G_{\ell,1}=0$;\\
\step\steplabel{feld02}
\>\>\>  Repeat $h=\ln(15M'k)/\ln(4/3)$ times\\
\step\steplabel{feld03}
\>\>\>\> Choose $b\in U$; \\
\>\>\>\>\> If $f(b_{Y_{\ell,0}}\circ b_{Y_{\ell,1}}\circ v^{(\ell)}_{\overline{X_\ell}})\not= f(\overline{b_{Y_{\ell,0}}}\circ b_{Y_{\ell,1}}\circ v^{(\ell)}_{\overline{X_\ell}})$ then $G_{\ell,0}\gets G_{\ell,0}+1$\\
\>\>\>\>\> If $f(b_{Y_{\ell,1}}\circ b_{Y_{\ell,0}}\circ v^{(\ell)}_{\overline{X_\ell}})\not= f(\overline{b_{Y_{\ell,1}}}\circ b_{Y_{\ell,0}}\circ v^{(\ell)}_{\overline{X_\ell}})$ then $G_{\ell,1}\gets G_{\ell,1}+1$\\
\step\steplabel{GGG}
\>\>\>If ($\{G_{\ell,0},G_{\ell,1}\}\not=\{0,h\}$) then output ``reject'' and halt\\
\step\steplabel{Gl}
\>\>\> If $G_{\ell,0}=h$ then $z \gets z\circ w_{X_\ell}$ else $z \gets z\circ \overline{w_{X_\ell}}$\\
\step  \>  Choose $u\in {\cal D}$\\
\step\steplabel{Finn}
\>  If $f(u_X\circ 0_{\overline{X}})\not= f((u_X+z_X)\circ 0_{\overline{X}})$ then output ``reject'' and halt.\\
\step\steplabel{accept}
Output ``accept''
  \end{tabbing}
  \end{minipage}}}
  \end{center}
	\caption{A two-sided distribution-free adaptive algorithm for $\epsilon$-testing $k$-junta.}
	\label{A3}
	\end{figure}

\subsection{The algorithm for $k$-Junta}
In this subsection, we show that if the target function $f$ is $k$-junta then the algorithm accepts with probability at least $2/3$.

We first prove
\begin{lemma}\label{dist} Consider steps~\ref{par1}-\ref{par2} in the algorithm. If $f$ is a $k$-junta then, with probability at least $2/3$, for each $i\in [r]$, the set $x(X_i)=\{x_j|j\in X_i\}$ contains at most one relevant variable of $f$.
\end{lemma}
\begin{proof} Let $x_{i_1}$ and $x_{i_2}$ be two relevant variables in $f$.
The probability that $x_{i_1}$ and $x_{i_2}$ are in the same set is equal to $1/r$.
By the union bound, it follows that the probability that some relevant variables $x_{i_1}$ and $x_{i_2}$ in $f$ are in the same set is at most ${k\choose 2}/r\le 1/3$.
\end{proof}

We now show that w.h.p. the algorithm reaches the final test in the algorithm
\begin{lemma}\label{kjun} If $f$ is $k$-junta and each $x(X_i)$ contains at most one relevant variable of $f$ then
\begin{enumerate}
\item Each $x(X_i)$, $i\in I$, contains exactly one relevant variable.
\item The algorithm reaches step~\ref{RepF}
\end{enumerate}
\end{lemma}
\begin{proof}
By Lemma~\ref{BiSe} and steps~\ref{con1}-\ref{Finddd}, for $\ell\in I$, $f(v^{(\ell)})\not=f(0_{X_\ell}\circ v_{\overline{X_\ell}}^{(\ell)})$ and therefore $x(X_\ell)$ contains exactly one relevant variable.  Thus, for every $\ell\in I$, $f(x_{X_\ell}\circ v_{\overline{X_\ell}}^{(\ell)})$ is a literal.

If the algorithm does not reach step~\ref{RepF}, then it either halts in step~\ref{Rej}, \ref{Rej2} or \ref{ConE}. If it halts in step~\ref{Rej} then $|I|>k$ and therefore, by Lemma~\ref{trivial01}, $f$ contains more than $k$ relevant variables and then it is not $k$-Junta. If it halts in step~\ref{Rej2} then, by Lemma~\ref{OneSide}, for some $X_\ell$, $\ell\in I$, $f(x_{X_\ell}\circ v^{(\ell)}_{\overline{X_\ell}})$ is not $1$-Junta (literal or constant function) and therefore $X_\ell$ contains at least two relevant variables. If it halts in step~\ref{ConE}, then $f(b_{X_\ell}\circ v^{(\ell)}_{\overline{X_\ell}})=f(\overline{b_{X_\ell}}\circ v^{(\ell)}_{\overline{X_\ell}})$ and then $f(x_{X_\ell}\circ v^{(\ell)}_{\overline{X_\ell}})$ is not a literal. In all cases we get a contradiction.
\end{proof}

We now give two Lemmas that show that, with probability at least $2/3$, the algorithm accepts $k$-junta.
\begin{lemma}\label{kjun2} If $f$ is $k$-Junta and each $x(X_i)$ contains at most one relevant variable of $f$ then the algorithm outputs ``accept''.
\end{lemma}
\begin{proof} By Lemma~\ref{kjun}, the algorithm reaches step~\ref{RepF}. We now show that it reaches step~\ref{accept}. Now we need to show that the algorithm does not halt in step~\ref{GGG} or \ref{Finn}.

Since $Y_{\ell,0}, Y_{\ell,1}$ is a partition of $X_\ell$, $\ell\in I$ and $X_\ell$ contains exactly one relevant variable in $x(X_\ell)$ of $f$, this variable is either in $x(Y_{\ell,0})$ or in $x(Y_{\ell,1})$ but not in both. Suppose w.l.o.g. it is in $x(Y_{\ell,0})$ and not in $x(Y_{\ell,1})$. Then $f(x_{Y_{\ell,0}}\circ b_{Y_{\ell,1}}\circ v^{(\ell)}_{\overline{X_\ell}})$ is a literal and $f(x_{Y_{\ell,1}}\circ b_{Y_{\ell,0}}\circ v^{(\ell)}_{\overline{X_\ell}})$ is a constant function. This implies that for any $b$, $f(b_{Y_{\ell,0}}\circ b_{Y_{\ell,1}}\circ v^{(\ell)}_{\overline{X_\ell}})\not= f(\overline{b_{Y_{\ell,0}}}\circ b_{Y_{\ell,1}}\circ v^{(\ell)}_{\overline{X_\ell}})$ and $f(b_{Y_{\ell,1}}\circ b_{Y_{\ell,0}}\circ v^{(\ell)}_{\overline{X_\ell}})= f(\overline{b_{Y_{\ell,1}}}\circ b_{Y_{\ell,0}}\circ v^{(\ell)}_{\overline{X_\ell}})$. Therefore, $G_{\ell,0}=h$ and $G_{\ell,1}=0$. Thus the algorithm does not halt in step~\ref{GGG}.

Now for every $X_\ell$, $\ell\in I$, let $\tau(\ell)\in X_\ell$ be such that $f(x_{X_\ell}\circ v^{(\ell)}_{\overline{X_\ell}})\in\{x_{\tau(\ell)},\overline{x_{\tau(\ell)}}\}$. If $\tau(\ell)\in Y_{\ell,0}$ then $G_{\ell,0}=h$ and then by step~\ref{Gl}, $z_{\tau(\ell)}=w_{\tau(\ell)}=0$. If $\tau(\ell)\in Y_{\ell,1}$ then $G_{\ell,1}=h$ and then $z_{\tau(\ell)}=\overline{w_{\tau(\ell)}}=0$. Therefore for every relevant variable $x_{\tau(\ell)}$ in $\hat f=f(x_X\circ 0_{\overline{X}})$ we have $z_{\tau(\ell)}=0$ which implies that $f(u_X\circ 0_{\overline{X}})=f((u_X+z_X)\circ 0_{\overline{X}})$ and therefore the algorithm does not halt in step~\ref{Finn}.
\end{proof}

\begin{lemma} If $f$ is $k$-Junta then the algorithm outputs ``accept'' with probability at least $2/3$ .
\end{lemma}
\begin{proof} The result follows from Lemma~\ref{dist} and Lemma~\ref{kjun2}.
\end{proof}

\subsection{The Algorithm for $\epsilon$-Far Functions}
In this subsection, we prove that if $f$ is $\epsilon$-far from every $k$-junta then the algorithm rejects with probability at least $2/3$.

The first lemma shows that, w.h.p., $f(u_X\circ 0_{\overline{X}})$ is $\epsilon/2$-close to $f$.
\begin{lemma}\label{cloose}
If the algorithm reaches step~\ref{LitT} then $t(X)=2\ln(15k)/\epsilon$ and $|I|\le k$. If $$\Pr_{u\in {\cal D}}[f(u_X\circ 0_{\overline{X}})\not= f(u)]\ge \epsilon/2$$
then the algorithm reaches step~\ref{LitT} with probability at most $1/15$.
\end{lemma}
\begin{proof}
The algorithm does not reaches step~\ref{LitT} if and only if it halts in step~\ref{Rej} and then $|I|>k$. The size of $I$ is increased by one each time the condition, $f(u_X\circ 0_{\overline{X}})\not=f(u)$, in step~\ref{con1}, is true. Therefore, if the algorithm reaches step~\ref{LitT} then the condition in step~\ref{con1} was true at most $k$ times and $|I|\le k$. Then steps~\ref{Find}-\ref{tx0} are executed at most $k$ times. Thus, $t()$ is updated to $0$ at most $k$ times.
The loop \ref{Cho}-\ref{EndRep} is repeated $M$ times and $t()$ is updated to $0$ at most $k$ times and therefore there is $X$ for which $t(X)=M/k=2\ln(15k)/\epsilon$. This implies that when the algorithm reaches step~\ref{LitT}, we have $t(X)=2\ln(15k)/\epsilon$.

The probability that the algorithm reaches step~\ref{LitT} with $\Pr_{u\in {\cal D}}[f(u_X\circ 0_{\overline{X}})\not= f(u)]> \epsilon/2$ is the probability that for one (of the at most $k$) $X'$, $\Pr_{u\in {\cal D}}[f(u_{X'}\circ 0_{\overline{X'}})\not= f(u)]> \epsilon/2$ and $t(X')=2\ln(15k)/\epsilon$. By the union bound, this probability is less than
$$k\left( 1-\frac{\epsilon}{2}\right)^{2\ln (15k)/\epsilon}=\frac{1}{15}.$$
\end{proof}

In the following lemma we show that, w.h.p, each $f(x_{X_\ell}\circ v_{\overline{X_\ell}}^{(\ell)})$ is close to a literal.
\begin{lemma}\label{closelit} Consider steps~\ref{LitT}-\ref{Rej2}. If for some $\ell\in I$, $f(x_{X_\ell}\circ v_{\overline{X_\ell}}^{(\ell)})$ is $(1/30)$-far from every literal with respect to the uniform distribution then, with probability at least $1-(2/15)$, the algorithm rejects.
\end{lemma}
\begin{proof} If $f(x_{X_\ell}\circ v_{\overline{X_\ell}}^{(\ell)})$ is $(1/30)$-far from every literal with respect to the uniform distribution then it is either (case 1) $(1/30)$-far from every $1$-Junta (literal or constant) or (case 2) $(1/30)$-far from every literal and $(1/30)$-close to $0$-Junta. In case 1, by Lemma~\ref{OneSide}, with probability at least $1-(1/15)$, ${\bf UniformJunta}$ $(f(x_{X_\ell}\circ v_{\overline{X_\ell}}^{(\ell)}),1,1/30,1/15)$ $=$ ``reject'' and then the algorithm rejects. In case 2, if $f(x_{X_\ell}\circ v_{\overline{X_\ell}}^{(\ell)})$ is $1/30$-close to some $0$-Junta then it is either $(1/30)$-close to $0$ or $(1/30)$-close to $1$. Suppose it is $(1/30)$-close to $0$. Let $b$ be a random uniform string generated in steps~\ref{ConB}. Then $\overline{b}$ is random uniform and for $g(x)=f(x_{X_\ell}\circ v^{(\ell)}_{\overline{X_\ell}})$ we have
\begin{eqnarray*}
\Pr[\mbox{The algorithm does not reject}]&=&
\Pr\left[g(b)\not=g(\overline{b})\right]\\
&=&\Pr[g(b)=1\wedge g(\overline{b})=0]+\Pr[g(b)=0\wedge g(\overline{b})=1]\\
&\le&\Pr[g(b)=1]+\Pr[g(\overline{b})=1]\\
&\le&\frac{1}{15}.
\end{eqnarray*}
By the union bound the result follows.
\end{proof}

In the next lemma we prove that, w.h.p, the string $z$ generated in steps~\ref{feld}-\ref{Gl} satisfies $z_J=0$ where $x(J)$ are relevant variables of $f(u_X\circ 0_{\overline{X}})$.
\begin{lemma}\label{FinTest} Consider steps~\ref{feld}-\ref{Gl}. If for every $\ell\in I$ the function $f(x_{X_\ell}\circ v^{(\ell)}_{\overline{X_\ell}})$ is $(1/30)$-close to a literal in $\{x_{\tau(\ell)},\bar{x}_{\tau(\ell)}\}$ with respect to the uniform distribution, where $\tau(\ell)\in X_\ell$,
and $\{G_{\ell,0},G_{\ell,1}\}=\{0,h\}$ then, with probability at least $1-k(3/4)^h$, we have:
For every $\ell\in I$,
$z_{\tau(\ell)}=0$.
\end{lemma}
\begin{proof} Fix some $\ell$. Suppose $f(x_{X_\ell}\circ v^{(\ell)}_{\overline{X_\ell}})$ is $(1/30)$-close to $x_{\tau(\ell)}$ with respect to the uniform distribution. The case when it is $(1/30)$-close to $\overline{{x}_{\tau(\ell)}}$ is similar. Since $X_\ell=Y_{\ell,0}\cup Y_{\ell,1}$ and $Y_{\ell,0}\cap Y_{\ell,1}=\emptyset$ we have that $\tau(\ell)\in Y_{\ell,0}$ or $\tau(\ell)\in Y_{\ell,1}$, but not both. Suppose $\tau(\ell)\in Y_{\ell,0}$. The case where $\tau(\ell)\in Y_{\ell,1}$ is similar. Define the random variable $Z(x_{X_\ell})=1$ if $f(x_{X_\ell}\circ v^{(\ell)}_{\overline{X_\ell}})\not=x_{\tau(\ell)}$ and $Z(x_{X_\ell})=0$ otherwise. Then
$$\E_{x_{X_\ell}\in U}[Z(x_{X_\ell})]\le \frac{1}{30}.$$ Therefore
$$\E_{x_{Y_{\ell,1}}\in U}\E_{x_{Y_{\ell,0}}\in U}[Z(x_{Y_{\ell,0}}\circ x_{Y_{\ell,1}})]\le \frac{1}{30}$$ and by Markov's bound
$$\Pr_{x_{Y_{\ell,1}}\in U}\left[ \E_{x_{Y_{\ell,0}}\in U}[Z(x_{Y_{\ell,0}}\circ x_{Y_{\ell,1}})]\ge \frac{2}{15}\right]\le \frac{1}{4}.$$
That is, for a random uniform string $b\in \{0,1\}^n$, with probability at least $3/4$, $f(x_{Y_{\ell,0}}\circ b_{Y_{\ell,1}}\circ v^{(\ell)}_{\overline{X_\ell}})$ is $(2/15)$-close to $x_{\tau(\ell)}$ with respect to the uniform distribution. Now, given that $f(x_{Y_{\ell,0}}\circ b_{Y_{\ell,1}}\circ v^{(\ell)}_{\overline{X_\ell}})$ is $(2/15)$-close to $x_{\tau(\ell)}$ with respect to the uniform distribution the probability that $G_{\ell,0}=0$ is the probability that $f(b_{Y_{\ell,0}}\circ b_{Y_{\ell,1}}\circ v^{(\ell)}_{\overline{X_\ell}})= f(\overline{b_{Y_{\ell,0}}}\circ b_{Y_{\ell,1}}\circ v^{(\ell)}_{\overline{X_\ell}})$ for $h$ random uniform strings $b\in \{0,1\}^n$. Let $b^{(1)},\ldots,b^{(h)}$ be $h$ random uniform strings in $\{0,1\}^n$, $V(b)$ be the event $f(b_{Y_{\ell,0}}\circ b_{Y_{\ell,1}}\circ v^{(\ell)}_{\overline{X_\ell}})= f(\overline{b_{Y_{\ell,0}}}\circ b_{Y_{\ell,1}}\circ v^{(\ell)}_{\overline{X_\ell}})$ and $A$ the event that $f(x_{Y_{\ell,0}}\circ b_{Y_{\ell,1}}\circ v^{(\ell)}_{\overline{X_\ell}})$ is $(2/15)$-close to $x_{\tau(\ell)}$ with respect to the uniform distribution. Let $g(x_{Y_{\ell,0}})=f(x_{Y_{\ell,0}}\circ b_{Y_{\ell,1}}\circ v^{(\ell)}_{\overline{X_\ell}})$. Then
\begin{eqnarray*}
\Pr[V(b)|A]&=& \Pr[g(b_{Y_{\ell,0}})=g(\overline{b_{Y_{\ell,0}}})|A]\\
&=&\Pr[(g(b_{Y_{\ell,0}})=b_{\tau(\ell)}\wedge g(\overline{b_{Y_{\ell,0}}})=b_{\tau(\ell)})
\vee (g(b_{Y_{\ell,0}})=\overline{b_{\tau(\ell)}}\wedge g(\overline{b_{Y_{\ell,0}}})=\overline{b_{\tau(\ell)}})|A]\\
&\le& \Pr[g(\overline{b_{Y_{\ell,0}}})\not=\overline{b_{\tau(\ell)}}
\vee g({b_{Y_{\ell,0}}})\not={b_{\tau(\ell)}})|A]\\
&\le& \Pr[g(\overline{b_{Y_{\ell,0}}})\not=\overline{b_{\tau(\ell)}}|A]
+\Pr[g({b_{Y_{\ell,0}}})\not={b_{\tau(\ell)}})|A]\le \frac{4}{15}.
\end{eqnarray*}
Since $\tau(\ell)\in Y_{\ell,0}$, we have $w_{\tau(\ell)}=0$. Therefore, by step~\ref{Gl} and since $\tau(\ell)\in X_\ell$,
\begin{eqnarray*}
\Pr[z_{\tau(\ell)}=1]&=&\Pr[G_{\ell,0}=0\wedge G_{\ell,1}=h]\\
&\le&\Pr[G_{\ell,0}=0]=\Pr[(\forall j\in[h]) V(b^{(j)})]\\
&=& (\Pr[V(b)])^h
\le \left( \Pr[V(b)|A]+\Pr[\overline{A}]\right)^h
\le (4/15+1/4)^h\le (3/4)^h
\end{eqnarray*}
Therefore, the probability that $z_{\tau(\ell)}=1$ for some $\ell\in I$ is at most $k(3/4)^h$.
\end{proof}

We now show that w.h.p the algorithm reject if $f$ is $\epsilon$-far from every $k$-junta

\begin{lemma}\label{kll} If $f$ is $\epsilon$-far from every $k$-junta with respect to $\D$ then, with probability at least $2/3$, the algorithm outputs ``reject''.
\end{lemma}
\begin{proof} If the algorithm stops in step~\ref{Rej} then we are done. Therefore we may assume that
\begin{eqnarray}\label{eq1}
|I|\le k.
\end{eqnarray}
By Lemma~\ref{cloose}, if $\Pr_{u\in {\cal D}}[f(u_X\circ 0_{\overline{X}})\not= f(u)]\ge \epsilon/2$ then, with probability at most $1/15$, the algorithm reaches step~\ref{LitT}. So we may assume that (failure probability $1/15$)
\begin{eqnarray}\label{eq2}
\Pr_{u\in {\cal D}}[f(u_X\circ 0_{\overline{X}})\not= f(u)]\le \epsilon/2.
\end{eqnarray}
Since $f$ is $\epsilon$-far from every $k$-junta with respect to ${\cal D}$ and $f(x_X\circ 0_{\overline{X}})$ is $\epsilon/2$-close to $f$ with respect to ${\cal D}$ we have $f(x_X\circ 0_{\overline{X}})$ is $(\epsilon/2)$-far from every $k$-junta with respect to ${\cal D}$. Therefore, by Lemma~\ref{disD},
\begin{eqnarray}\label{hhhh}
\underset{u\in{\cal D},y\in U}{\Pr}[f(u_X\circ 0_{\overline{X}})= f(u_I\circ y_{X\backslash I}\circ 0_{\overline{X}})]\ge 1-\frac{\epsilon}{2}.
\end{eqnarray}

By Lemma~\ref{closelit}, if some $f(x_{X_\ell}\circ v_{\overline{X_\ell}}^{(\ell)})$ is $(1/30)$-far from any literal with respect to the uniform distribution then, with probability at least $1-(2/15)$, the algorithm rejects. So we may assume (failure probability $2/15$) that every $f(x_{X_\ell}\circ v_{\overline{X_\ell}}^{(\ell)})$ is $(1/30)$-close to some $x_{\tau(\ell)}$ or $\overline{x_{\tau(\ell)}}$ with respect to the uniform distribution, where $\tau(\ell)\in X_\ell$.

Let $z^{(1)},\ldots,z^{(M')}$ be the strings generated in step~\ref{Gl}. By Lemma~\ref{FinTest}, with probability at least $1-M'k(3/4)^h\ge 1-(1/15)$, every $z^{(i)}$ generated in step~\ref{Gl} satisfies $z^{(i)}_{\tau(\ell)}=0$ for all $\ell\in I$.
Also, since the distribution of $w_{X_\ell}$ and $\overline{w_{X_\ell}}$ is uniform, the distribution of $z^{(i)}_{X\backslash I}$ and $u_{X\backslash I}+z^{(i)}_{X\backslash I}$ is uniform. We now assume (failure probability $1/15$) that $z^{(i)}_I=0$ for all $i$. Therefore, by~(\ref{hhhh}),
\begin{eqnarray*}
\underset{u\in{\cal D},z^{(i)}_{X\backslash I}\in U}{\Pr}[(\forall i)f(u_X\circ 0_{\overline{X}})&= &f((u_X+z^{(i)}_X)\circ 0_{\overline{X}})]\\ &=&\left(\underset{u\in{\cal D},z^{(1)}_{X\backslash I}\in U}{\Pr}[f(u_X\circ 0_{\overline{X}})= f((u_X+z^{(1)}_X)\circ 0_{\overline{X}})]\right)^{M'}\\
&=&\left(\underset{u\in{\cal D},y\in U}{\Pr}[f(u_X\circ 0_{\overline{X}})= f(u_I\circ y_{X\backslash I}\circ 0_{\overline{X}})]\right)^{M'}\\
&\le & (1-\epsilon/2)^{M'}\le \frac{1}{15}.
\end{eqnarray*}
Therefore, the failure probability of an output ``reject'' is at most $1/15+2/15+1/15+1/15=1/3$.
\end{proof}

\subsection{The Query Complexity of the Algorithm}
In this section we show that
\begin{lemma} The query complexity of the algorithm is
$$\tilde O\left(\frac{k}{\epsilon}\right).$$
\end{lemma}
\begin{proof} The condition in step~\ref{con1} requires two queries and is executed at most $M=2k\ln(15k)/\epsilon$ times. This is $2M=O((k\log k)/\epsilon)$ queries. Steps~\ref{Find} is executed at most $k+1$ times. This is because each time it is executed, the value of $|I|$ is increased by one, and when $|I|=k+1$ the algorithm rejects. By Lemma~\ref{BiSe}, to find a new relevant set the algorithm makes $O(\log r)=O(\log k)$ queries. This is $O(k\log k)$ queries.
Steps~\ref{Uni} and \ref{ConE} are executed $|I|\le k$ times, and by Lemma~\ref{OneSide}, the total number of queries made is $O(1/(1/30)\log(15))k+2k=O(k)$.

The final test in the algorithm is repeated $M'=(2\ln 15)/\epsilon$ times (step~\ref{RepF}) and each time, and for each $\ell\in I$, (step~\ref{feld01}) it repeats $h$ times (step~\ref{feld02}) two conditions that takes $2$ queries each (step~\ref{feld03}). This takes $4M'kh=O((k/\epsilon)\ln(k/\epsilon))$ queries. The number of queries in step~\ref{Finn} is $2M'=O(1/\epsilon)$. Therefore the total number of queries is
 $$O\left(\frac{k}{\epsilon}\ln \frac{k}{\epsilon}\right).$$
\end{proof}

\section{Open Problems}
In this paper we proved that for any $\epsilon>0$, there is a two-sided distribution-free adaptive algorithm for $\epsilon$-testing $k$-junta that makes $\tilde O(k/\epsilon)$ queries. It is also interesting to find a one-sided distribution-free adaptive algorithm with such query complexity.

Chen et al.~\cite{LiuCSSX18} proved the lower bound $\Omega(2^{k/3})$ for any non-adaptive (one round) algorithm. What is the minimal number rounds one needs to get $poly(k/\epsilon)$ query complexity? Can $O(1)$-round algorithms solve the problem with $poly(k/\epsilon)$ queries?

In the uniform distribution framework, where the distance between two functions is measured with respect to the uniform distribution Blais in~\cite{Blais08} gave a non-adaptive algorithm that makes $\tilde{O}(k^{3/2})/\epsilon$ queries and in~\cite{Blais09} an adaptive algorithm that makes $O(k\log k+k/\epsilon)$ queries.
On the lower bounds side, Sa\u{g}lam in~ \cite{Saglam18} gave an $\Omega(k\log k)$ lower bound for adaptive testing and Chen et al.~\cite{ChenSTWX17} gave an $\tilde{\Omega}(k^{3/2})/\epsilon$ lower bound for the non-adaptive testing. Thus in both the adaptive and non-adaptive uniform distribution settings, the query complexity of $k$-junta testing has now been pinned down to within logarithmic factors. It is interesting to study $O(1)$-round algorithms. For example, what is the query complexity for $2$-round algorithm.

$$ $$
\noindent
{\bf Acknowledgment.} We would like to thank Xi Chen for reading the early version of the
paper and for verifying the correctness of the algorithm.

\bibliography{TestingRef}

\begin{thebibliography}{10}

\bibitem{AlonKKLR05}
Noga Alon, Tali Kaufman, Michael Krivelevich, Simon Litsyn, and Dana Ron.
\newblock Testing reed-muller codes.
\newblock {\em {IEEE} Trans. Information Theory}, 51(11):4032--4039, 2005.
\newblock URL: \url{https://doi.org/10.1109/TIT.2005.856958}, \href
  {http://dx.doi.org/10.1109/TIT.2005.856958}
  {\path{doi:10.1109/TIT.2005.856958}}.

\bibitem{BaleshzarMPR16}
Roksana Baleshzar, Meiram Murzabulatov, Ramesh Krishnan~S. Pallavoor, and Sofya
  Raskhodnikova.
\newblock Testing unateness of real-valued functions.
\newblock {\em CoRR}, abs/1608.07652, 2016.
\newblock URL: \url{http://arxiv.org/abs/1608.07652}, \href
  {http://arxiv.org/abs/1608.07652} {\path{arXiv:1608.07652}}.

\bibitem{BelovsB16}
Aleksandrs Belovs and Eric Blais.
\newblock A polynomial lower bound for testing monotonicity.
\newblock In {\em Proceedings of the 48th Annual {ACM} {SIGACT} Symposium on
  Theory of Computing, {STOC} 2016, Cambridge, MA, USA, June 18-21, 2016},
  pages 1021--1032, 2016.
\newblock URL: \url{https://doi.org/10.1145/2897518.2897567}, \href
  {http://dx.doi.org/10.1145/2897518.2897567}
  {\path{doi:10.1145/2897518.2897567}}.

\bibitem{BhattacharyyaKSSZ10}
Arnab Bhattacharyya, Swastik Kopparty, Grant Schoenebeck, Madhu Sudan, and
  David Zuckerman.
\newblock Optimal testing of reed-muller codes.
\newblock In {\em Property Testing - Current Research and Surveys}, pages
  269--275. 2010.
\newblock URL: \url{https://doi.org/10.1007/978-3-642-16367-8\_19}, \href
  {http://dx.doi.org/10.1007/978-3-642-16367-8\_19}
  {\path{doi:10.1007/978-3-642-16367-8\_19}}.

\bibitem{Blais08}
Eric Blais.
\newblock Improved bounds for testing juntas.
\newblock In {\em Approximation, Randomization and Combinatorial Optimization.
  Algorithms and Techniques, 11th International Workshop, {APPROX} 2008, and
  12th International Workshop, {RANDOM} 2008, Boston, MA, USA, August 25-27,
  2008. Proceedings}, pages 317--330, 2008.
\newblock URL: \url{https://doi.org/10.1007/978-3-540-85363-3\_26}, \href
  {http://dx.doi.org/10.1007/978-3-540-85363-3\_26}
  {\path{doi:10.1007/978-3-540-85363-3\_26}}.

\bibitem{Blais09}
Eric Blais.
\newblock Testing juntas nearly optimally.
\newblock In {\em Proceedings of the 41st Annual {ACM} Symposium on Theory of
  Computing, {STOC} 2009, Bethesda, MD, USA, May 31 - June 2, 2009}, pages
  151--158, 2009.
\newblock URL: \url{https://doi.org/10.1145/1536414.1536437}, \href
  {http://dx.doi.org/10.1145/1536414.1536437}
  {\path{doi:10.1145/1536414.1536437}}.

\bibitem{BlaisBM11}
Eric Blais, Joshua Brody, and Kevin Matulef.
\newblock Property testing lower bounds via communication complexity.
\newblock In {\em Proceedings of the 26th Annual {IEEE} Conference on
  Computational Complexity, {CCC} 2011, San Jose, California, USA, June 8-10,
  2011}, pages 210--220, 2011.
\newblock URL: \url{https://doi.org/10.1109/CCC.2011.31}, \href
  {http://dx.doi.org/10.1109/CCC.2011.31} {\path{doi:10.1109/CCC.2011.31}}.

\bibitem{BlaisK12}
Eric Blais and Daniel~M. Kane.
\newblock Tight bounds for testing k-linearity.
\newblock In {\em Approximation, Randomization, and Combinatorial Optimization.
  Algorithms and Techniques - 15th International Workshop, {APPROX} 2012, and
  16th International Workshop, {RANDOM} 2012, Cambridge, MA, USA, August 15-17,
  2012. Proceedings}, pages 435--446, 2012.
\newblock URL: \url{https://doi.org/10.1007/978-3-642-32512-0\_37}, \href
  {http://dx.doi.org/10.1007/978-3-642-32512-0\_37}
  {\path{doi:10.1007/978-3-642-32512-0\_37}}.

\bibitem{Blum03}
Avrim Blum.
\newblock Learning a function of r relevant variables.
\newblock In {\em Computational Learning Theory and Kernel Machines, 16th
  Annual Conference on Computational Learning Theory and 7th Kernel Workshop,
  COLT/Kernel 2003, Washington, DC, USA, August 24-27, 2003, Proceedings},
  pages 731--733, 2003.
\newblock URL: \url{https://doi.org/10.1007/978-3-540-45167-9\_54}, \href
  {http://dx.doi.org/10.1007/978-3-540-45167-9\_54}
  {\path{doi:10.1007/978-3-540-45167-9\_54}}.

\bibitem{BlumL97}
Avrim Blum and Pat Langley.
\newblock Selection of relevant features and examples in machine learning.
\newblock {\em Artif. Intell.}, 97(1-2):245--271, 1997.
\newblock URL: \url{https://doi.org/10.1016/S0004-3702(97)00063-5}, \href
  {http://dx.doi.org/10.1016/S0004-3702(97)00063-5}
  {\path{doi:10.1016/S0004-3702(97)00063-5}}.

\bibitem{BlumLR93}
Manuel Blum, Michael Luby, and Ronitt Rubinfeld.
\newblock Self-testing/correcting with applications to numerical problems.
\newblock {\em J. Comput. Syst. Sci.}, 47(3):549--595, 1993.
\newblock URL: \url{https://doi.org/10.1016/0022-0000(93)90044-W}, \href
  {http://dx.doi.org/10.1016/0022-0000(93)90044-W}
  {\path{doi:10.1016/0022-0000(93)90044-W}}.

\bibitem{BshoutyC18}
Nader~H. Bshouty and Areej Costa.
\newblock Exact learning of juntas from membership queries.
\newblock {\em Theor. Comput. Sci.}, 742:82--97, 2018.
\newblock URL: \url{https://doi.org/10.1016/j.tcs.2017.12.032}, \href
  {http://dx.doi.org/10.1016/j.tcs.2017.12.032}
  {\path{doi:10.1016/j.tcs.2017.12.032}}.

\bibitem{ChakrabartyS13}
Deeparnab Chakrabarty and C.~Seshadhri.
\newblock A $o(n)$ monotonicity tester for boolean functions over the
  hypercube.
\newblock In {\em Symposium on Theory of Computing Conference, STOC'13, Palo
  Alto, CA, USA, June 1-4, 2013}, pages 411--418, 2013.
\newblock URL: \url{https://doi.org/10.1145/2488608.2488660}, \href
  {http://dx.doi.org/10.1145/2488608.2488660}
  {\path{doi:10.1145/2488608.2488660}}.

\bibitem{ChakrabartyS16a}
Deeparnab Chakrabarty and C.~Seshadhri.
\newblock A $\tilde{O}(n)$ non-adaptive tester for unateness.
\newblock {\em CoRR}, abs/1608.06980, 2016.
\newblock URL: \url{http://arxiv.org/abs/1608.06980}, \href
  {http://arxiv.org/abs/1608.06980} {\path{arXiv:1608.06980}}.

\bibitem{ChenDST15}
Xi~Chen, Anindya De, Rocco~A. Servedio, and Li{-}Yang Tan.
\newblock Boolean function monotonicity testing requires (almost) $n^{1/2}$
  non-adaptive queries.
\newblock In {\em Proceedings of the Forty-Seventh Annual {ACM} on Symposium on
  Theory of Computing, {STOC} 2015, Portland, OR, USA, June 14-17, 2015}, pages
  519--528, 2015.
\newblock URL: \url{https://doi.org/10.1145/2746539.2746570}, \href
  {http://dx.doi.org/10.1145/2746539.2746570}
  {\path{doi:10.1145/2746539.2746570}}.

\bibitem{ChenST14}
Xi~Chen, Rocco~A. Servedio, and Li{-}Yang Tan.
\newblock New algorithms and lower bounds for monotonicity testing.
\newblock {\em CoRR}, abs/1412.5655, 2014.
\newblock URL: \url{http://arxiv.org/abs/1412.5655}, \href
  {http://arxiv.org/abs/1412.5655} {\path{arXiv:1412.5655}}.

\bibitem{ChenSTWX17}
Xi~Chen, Rocco~A. Servedio, Li{-}Yang Tan, Erik Waingarten, and Jinyu Xie.
\newblock Settling the query complexity of non-adaptive junta testing.
\newblock In {\em 32nd Computational Complexity Conference, {CCC} 2017, July
  6-9, 2017, Riga, Latvia}, pages 26:1--26:19, 2017.
\newblock URL: \url{https://doi.org/10.4230/LIPIcs.CCC.2017.26}, \href
  {http://dx.doi.org/10.4230/LIPIcs.CCC.2017.26}
  {\path{doi:10.4230/LIPIcs.CCC.2017.26}}.

\bibitem{ChenWX17}
Xi~Chen, Erik Waingarten, and Jinyu Xie.
\newblock Beyond talagrand functions: new lower bounds for testing monotonicity
  and unateness.
\newblock In {\em Proceedings of the 49th Annual {ACM} {SIGACT} Symposium on
  Theory of Computing, {STOC} 2017, Montreal, QC, Canada, June 19-23, 2017},
  pages 523--536, 2017.
\newblock URL: \url{https://doi.org/10.1145/3055399.3055461}, \href
  {http://dx.doi.org/10.1145/3055399.3055461}
  {\path{doi:10.1145/3055399.3055461}}.

\bibitem{ChenWX17b}
Xi~Chen, Erik Waingarten, and Jinyu Xie.
\newblock Boolean unateness testing with $\widetilde{O}(n^{3/4})$ adaptive
  queries.
\newblock In {\em 58th {IEEE} Annual Symposium on Foundations of Computer
  Science, {FOCS} 2017, Berkeley, CA, USA, October 15-17, 2017}, pages
  868--879, 2017.
\newblock URL: \url{https://doi.org/10.1109/FOCS.2017.85}, \href
  {http://dx.doi.org/10.1109/FOCS.2017.85} {\path{doi:10.1109/FOCS.2017.85}}.

\bibitem{ChenX16}
Xi~Chen and Jinyu Xie.
\newblock Tight bounds for the distribution-free testing of monotone
  conjunctions.
\newblock In {\em Proceedings of the Twenty-Seventh Annual {ACM-SIAM} Symposium
  on Discrete Algorithms, {SODA} 2016, Arlington, VA, USA, January 10-12,
  2016}, pages 54--71, 2016.
\newblock URL: \url{https://doi.org/10.1137/1.9781611974331.ch5}, \href
  {http://dx.doi.org/10.1137/1.9781611974331.ch5}
  {\path{doi:10.1137/1.9781611974331.ch5}}.

\bibitem{ChocklerG04}
Hana Chockler and Dan Gutfreund.
\newblock A lower bound for testing juntas.
\newblock {\em Inf. Process. Lett.}, 90(6):301--305, 2004.
\newblock URL: \url{https://doi.org/10.1016/j.ipl.2004.01.023}, \href
  {http://dx.doi.org/10.1016/j.ipl.2004.01.023}
  {\path{doi:10.1016/j.ipl.2004.01.023}}.

\bibitem{DiakonikolasLMORSW07}
Ilias Diakonikolas, Homin~K. Lee, Kevin Matulef, Krzysztof Onak, Ronitt
  Rubinfeld, Rocco~A. Servedio, and Andrew Wan.
\newblock Testing for concise representations.
\newblock In {\em 48th Annual {IEEE} Symposium on Foundations of Computer
  Science {(FOCS} 2007), October 20-23, 2007, Providence, RI, USA,
  Proceedings}, pages 549--558, 2007.
\newblock URL: \url{https://doi.org/10.1109/FOCS.2007.32}, \href
  {http://dx.doi.org/10.1109/FOCS.2007.32} {\path{doi:10.1109/FOCS.2007.32}}.

\bibitem{DolevR11}
Elya Dolev and Dana Ron.
\newblock Distribution-free testing for monomials with a sublinear number of
  queries.
\newblock {\em Theory of Computing}, 7(1):155--176, 2011.
\newblock URL: \url{https://doi.org/10.4086/toc.2011.v007a011}, \href
  {http://dx.doi.org/10.4086/toc.2011.v007a011}
  {\path{doi:10.4086/toc.2011.v007a011}}.

\bibitem{FischerKRSS04}
Eldar Fischer, Guy Kindler, Dana Ron, Shmuel Safra, and Alex Samorodnitsky.
\newblock Testing juntas.
\newblock {\em J. Comput. Syst. Sci.}, 68(4):753--787, 2004.
\newblock URL: \url{https://doi.org/10.1016/j.jcss.2003.11.004}, \href
  {http://dx.doi.org/10.1016/j.jcss.2003.11.004}
  {\path{doi:10.1016/j.jcss.2003.11.004}}.

\bibitem{FischerLNRRS02}
Eldar Fischer, Eric Lehman, Ilan Newman, Sofya Raskhodnikova, Ronitt Rubinfeld,
  and Alex Samorodnitsky.
\newblock Monotonicity testing over general poset domains.
\newblock In {\em Proceedings on 34th Annual {ACM} Symposium on Theory of
  Computing, May 19-21, 2002, Montr{\'{e}}al, Qu{\'{e}}bec, Canada}, pages
  474--483, 2002.
\newblock URL: \url{https://doi.org/10.1145/509907.509977}, \href
  {http://dx.doi.org/10.1145/509907.509977} {\path{doi:10.1145/509907.509977}}.

\bibitem{GlasnerS09}
Dana Glasner and Rocco~A. Servedio.
\newblock Distribution-free testing lower bound for basic boolean functions.
\newblock {\em Theory of Computing}, 5(1):191--216, 2009.
\newblock URL: \url{https://doi.org/10.4086/toc.2009.v005a010}, \href
  {http://dx.doi.org/10.4086/toc.2009.v005a010}
  {\path{doi:10.4086/toc.2009.v005a010}}.

\bibitem{GoldreichSurvey10}
Oded Goldreich, editor.
\newblock {\em Property Testing - Current Research and Surveys}, volume 6390 of
  {\em Lecture Notes in Computer Science}.
\newblock Springer, 2010.
\newblock URL: \url{https://doi.org/10.1007/978-3-642-16367-8}, \href
  {http://dx.doi.org/10.1007/978-3-642-16367-8}
  {\path{doi:10.1007/978-3-642-16367-8}}.

\bibitem{GoldreichGLRS00}
Oded Goldreich, Shafi Goldwasser, Eric Lehman, Dana Ron, and Alex
  Samorodnitsky.
\newblock Testing monotonicity.
\newblock {\em Combinatorica}, 20(3):301--337, 2000.
\newblock URL: \url{https://doi.org/10.1007/s004930070011}, \href
  {http://dx.doi.org/10.1007/s004930070011} {\path{doi:10.1007/s004930070011}}.

\bibitem{GoldreichGR98}
Oded Goldreich, Shafi Goldwasser, and Dana Ron.
\newblock Property testing and its connection to learning and approximation.
\newblock {\em J. {ACM}}, 45(4):653--750, 1998.
\newblock URL: \url{https://doi.org/10.1145/285055.285060}, \href
  {http://dx.doi.org/10.1145/285055.285060} {\path{doi:10.1145/285055.285060}}.

\bibitem{GopalanOSSW11}
Parikshit Gopalan, Ryan O'Donnell, Rocco~A. Servedio, Amir Shpilka, and Karl
  Wimmer.
\newblock Testing fourier dimensionality and sparsity.
\newblock {\em {SIAM} J. Comput.}, 40(4):1075--1100, 2011.
\newblock URL: \url{https://doi.org/10.1137/100785429}, \href
  {http://dx.doi.org/10.1137/100785429} {\path{doi:10.1137/100785429}}.

\bibitem{GuijarroTT98}
David Guijarro, Jun Tarui, and Tatsuie Tsukiji.
\newblock Finding relevant variables in {PAC} model with membership queries.
\newblock In {\em Algorithmic Learning Theory, 10th International Conference,
  {ALT} '99, Tokyo, Japan, December 6-8, 1999, Proceedings}, page 313, 1999.
\newblock URL: \url{https://doi.org/10.1007/3-540-46769-6\_26}, \href
  {http://dx.doi.org/10.1007/3-540-46769-6\_26}
  {\path{doi:10.1007/3-540-46769-6\_26}}.

\bibitem{HalevyK07}
Shirley Halevy and Eyal Kushilevitz.
\newblock Distribution-free property-testing.
\newblock {\em {SIAM} J. Comput.}, 37(4):1107--1138, 2007.
\newblock URL: \url{https://doi.org/10.1137/050645804}, \href
  {http://dx.doi.org/10.1137/050645804} {\path{doi:10.1137/050645804}}.

\bibitem{KhotMS15}
Subhash Khot, Dor Minzer, and Muli Safra.
\newblock On monotonicity testing and boolean isoperimetric type theorems.
\newblock In {\em {IEEE} 56th Annual Symposium on Foundations of Computer
  Science, {FOCS} 2015, Berkeley, CA, USA, 17-20 October, 2015}, pages 52--58,
  2015.
\newblock URL: \url{https://doi.org/10.1109/FOCS.2015.13}, \href
  {http://dx.doi.org/10.1109/FOCS.2015.13} {\path{doi:10.1109/FOCS.2015.13}}.

\bibitem{KhotS16}
Subhash Khot and Igor Shinkar.
\newblock An $\widetilde{O}(n)$ queries adaptive tester for unateness.
\newblock {\em CoRR}, abs/1608.02451, 2016.
\newblock URL: \url{http://arxiv.org/abs/1608.02451}, \href
  {http://arxiv.org/abs/1608.02451} {\path{arXiv:1608.02451}}.

\bibitem{LiptonMMV05}
Richard~J. Lipton, Evangelos Markakis, Aranyak Mehta, and Nisheeth~K. Vishnoi.
\newblock On the fourier spectrum of symmetric boolean functions with
  applications to learning symmetric juntas.
\newblock In {\em 20th Annual {IEEE} Conference on Computational Complexity
  {(CCC} 2005), 11-15 June 2005, San Jose, CA, {USA}}, pages 112--119, 2005.
\newblock URL: \url{https://doi.org/10.1109/CCC.2005.19}, \href
  {http://dx.doi.org/10.1109/CCC.2005.19} {\path{doi:10.1109/CCC.2005.19}}.

\bibitem{LiuCSSX18}
Zhengyang Liu, Xi~Chen, Rocco~A. Servedio, Ying Sheng, and Jinyu Xie.
\newblock Distribution-free junta testing.
\newblock In {\em Proceedings of the 50th Annual {ACM} {SIGACT} Symposium on
  Theory of Computing, {STOC} 2018, Los Angeles, CA, USA, June 25-29, 2018},
  pages 749--759, 2018.
\newblock URL: \url{https://doi.org/10.1145/3188745.3188842}, \href
  {http://dx.doi.org/10.1145/3188745.3188842}
  {\path{doi:10.1145/3188745.3188842}}.

\bibitem{MatulefORS09}
Kevin Matulef, Ryan O'Donnell, Ronitt Rubinfeld, and Rocco~A. Servedio.
\newblock Testing {\(\pm\)}1-weight halfspace.
\newblock In {\em Approximation, Randomization, and Combinatorial Optimization.
  Algorithms and Techniques, 12th International Workshop, {APPROX} 2009, and
  13th International Workshop, {RANDOM} 2009, Berkeley, CA, USA, August 21-23,
  2009. Proceedings}, pages 646--657, 2009.
\newblock URL: \url{https://doi.org/10.1007/978-3-642-03685-9\_48}, \href
  {http://dx.doi.org/10.1007/978-3-642-03685-9\_48}
  {\path{doi:10.1007/978-3-642-03685-9\_48}}.

\bibitem{MatulefORS10}
Kevin Matulef, Ryan O'Donnell, Ronitt Rubinfeld, and Rocco~A. Servedio.
\newblock Testing halfspaces.
\newblock {\em {SIAM} J. Comput.}, 39(5):2004--2047, 2010.
\newblock URL: \url{https://doi.org/10.1137/070707890}, \href
  {http://dx.doi.org/10.1137/070707890} {\path{doi:10.1137/070707890}}.

\bibitem{MosselOS04}
Elchanan Mossel, Ryan O'Donnell, and Rocco~A. Servedio.
\newblock Learning functions of k relevant variables.
\newblock {\em J. Comput. Syst. Sci.}, 69(3):421--434, 2004.
\newblock URL: \url{https://doi.org/10.1016/j.jcss.2004.04.002}, \href
  {http://dx.doi.org/10.1016/j.jcss.2004.04.002}
  {\path{doi:10.1016/j.jcss.2004.04.002}}.

\bibitem{ParnasRS02}
Michal Parnas, Dana Ron, and Alex Samorodnitsky.
\newblock Testing basic boolean formulae.
\newblock {\em {SIAM} J. Discrete Math.}, 16(1):20--46, 2002.
\newblock URL: \url{http://epubs.siam.org/sam-bin/dbq/article/40744}.

\bibitem{Ron08}
Dana Ron.
\newblock Property testing: {A} learning theory perspective.
\newblock {\em Foundations and Trends in Machine Learning}, 1(3):307--402,
  2008.
\newblock URL: \url{https://doi.org/10.1561/2200000004}, \href
  {http://dx.doi.org/10.1561/2200000004} {\path{doi:10.1561/2200000004}}.

\bibitem{Ron09}
Dana Ron.
\newblock Algorithmic and analysis techniques in property testing.
\newblock {\em Foundations and Trends in Theoretical Computer Science},
  5(2):73--205, 2009.
\newblock URL: \url{https://doi.org/10.1561/0400000029}, \href
  {http://dx.doi.org/10.1561/0400000029} {\path{doi:10.1561/0400000029}}.

\bibitem{RubinfeldS96}
Ronitt Rubinfeld and Madhu Sudan.
\newblock Robust characterizations of polynomials with applications to program
  testing.
\newblock {\em {SIAM} J. Comput.}, 25(2):252--271, 1996.
\newblock URL: \url{https://doi.org/10.1137/S0097539793255151}, \href
  {http://dx.doi.org/10.1137/S0097539793255151}
  {\path{doi:10.1137/S0097539793255151}}.

\bibitem{Saglam18}
Mert Saglam.
\newblock Near log-convexity of measured heat in (discrete) time and
  consequences.
\newblock In {\em 59th {IEEE} Annual Symposium on Foundations of Computer
  Science, {FOCS} 2018, Paris, France, October 7-9, 2018}, pages 967--978,
  2018.
\newblock URL: \url{https://doi.org/10.1109/FOCS.2018.00095}, \href
  {http://dx.doi.org/10.1109/FOCS.2018.00095}
  {\path{doi:10.1109/FOCS.2018.00095}}.

\end{thebibliography}

\end{document}